\documentclass[preprint,10pt,a4paper]{elsarticle}

\usepackage{amscd,amsfonts,amsmath,amssymb}
\usepackage{mathrsfs}
\usepackage[letterpaper,hmargin=1.5in,vmargin=1.5in]{geometry}
\usepackage[all]{xy}
\newtheorem{theorem}{Theorem}
\newtheorem{definition}{Definition}
\newtheorem{lemma}{Lemma}

\newdefinition{remark}{Remark}
\newproof{proof}{Proof}
\newproof{pot}{Proof of Theorem \ref{th:2}}

\journal{Journal of Pure and Applied Algebra}

\begin{document}

\begin{frontmatter}



\title{Minimal Logarithmic Signatures for Sporadic Groups}


\author[H]{Haibo Hong\corref{cor2}\fnref{fn1}}
\ead{honghaibo1985@163.com }

\author[L]{Licheng Wang\corref{cor1} \fnref{fn2}}
\ead{wanglc@bupt.edu.cn }

\author[A]{Haseeb Ahmad\corref{cor2} \fnref{fn3}}
\ead{haseeb\_ @hotmail.com}

\author[J]{Jing Li\corref{cor3} \fnref{fn4}}
\ead{l19861986@126.com}

\author[Y]{Yixian Yang\corref{cor4} \fnref{fn5}}
\ead{yxyang@bupt.edu.cn}

\address{Information Security Center, State Key Laboratory of
Networking and Switching Technology, Beijing University of Posts and
Telecommunications, Beijing, 100876 P.R. China}



\begin{abstract}
As a special type of factorization of finite groups, logarithmic signature (LS) is used as the main component of cryptographic
keys for secret key cryptosystems such as PGM and public key
cryptosystems like $MST_1$, $MST_2$ and $MST_3$.
An LS with the shortest length is called a minimal logarithmic signature (MLS) and is
even desirable for cryptographic constructions.
The MLS conjecture  states that every finite simple group has
an MLS. Until now, the MLS conjecture has been proved true for some families of simple groups. In this paper, we will prove the existence of minimal logarithmic signatures for some sporadic groups. \\

\begin{keyword}

Sporadic groups \sep (Minimal) logarithmic signature\sep Stabilizer\sep Solvable groups

\end{keyword}
\end{abstract}
\end{frontmatter}


\section{Introduction}

%


Currently, most asymmetric cryptographic primitives are based on the perceived intractability of certain mathematical problems in very large finite abelian groups \cite{LMTW09}. Prominent hard problems consist of the problem of factoring large integers, the Discrete Logarithm
Problem (DLP) over a finite field $F_q$ or an elliptic curve, etc. However, due to quantum algorithms for integer factoring and solving the DLP, most known public-key systems will be insecure when quantum computers become practical.  
Therefore, it is an imminent work to design effective cryptographic schemes which can resist quantum attack. Actually, several attempts using problems from mathematical areas like group theory have been made and some available cryptographic schemes such as $MST_1$, $MST_2$ and $MST_3$, which take advantage of logarithmic signatures(LS)\cite{MST02,LT05,LMTW09,MST12}, have been devised successively.
Generally, the difficulty of the factorization problem for a given logarithmic signature in non-abelian groups can be used as the basis of the security of $MST_3$ cryptosystem.

In order to apply logarithmic signatures in some practical cryptographic
schemes effectively, the question of finding minimal logarithmic signatures (MLS) arises naturally. Especially,  LS-based cryptosystems using MLSs will have the minimal space complexity. Besides, from the mathematical point of view, an MLS for a group can help us understand the structure of the group better. However, the problem of existence of MLSs for finite groups should be taken into account firstly.
%

In fact, some effective work has been done in searching the MLSs for finite groups. In 2003, Vasco et al. \cite{GRS03} proved that the MLSs exist for all groups of order less than 175,560  . In 2004, Holmes \cite{H04} gave MLSs for sporadic groups $J_1$, $J_2$, $HS$, $M^cL$, $He$, $Co_3$, $Ru$ and $Suz$. In 2005, Lempken et al. \cite{LT05} built MLSs for the special linear groups $SL_n(q)$ and the projective special linear groups $PSL_n(q)$ when $\gcd(n, q - 1)\in \{1, 4, p \}$ where $p$ is a prime. They also show that, with a few exceptions, an MLS exists for all groups of order $\leq 10^{10}$ including five Sporadic groups -- $M_{11}$, $M_{12}$, $M_{22}$, $M_{23}$ and $M_{24}$.
 Recently, Nikhil, Nidhi and Magliveras \cite{NNM10,NN11} proposed MLSs for the groups $GL_n(q)$, $SL_n(q)$, $Sp_n(q)$ and $O_{2m}^{\pm}(q)$. Meanwhile, they put forward the MLS conjecture that says every finite simple group has an MLS. Once the conjecture comes true, as a direct consequence of $\mathrm{Jordan-H\ddot{o}lder}$ Theorem, then any finite groups also have MLSs. While until now, the MLS conjecture remains open. So in this sense, it is meaningful for us to continue the work. In this paper, we go on considering the existence of MLSs for the remained 13 types of sporadic groups.

Our basic technique is as follows: Given a permutation representation of a group $G$, identify a point $p$ so that its stabilizer $G_p$ can be factored
through an minimal length logarithmic signature and such
that there exists a complete set of representatives of $G$
modulo $G_p$ which moves $p$ cyclically. More generally, we can also divide $G$
into $G_p$ and some cyclic sets (subgroups), if every part has an MLS, then $G$ also has an MLS.

For each kind of sporadic groups, the proposed MLS has the similar structures $[H,G_w]$, where $G_w$ is the stabilizer of corresponding sporadic groups $G$; $H$ is a product of some appropriate subgroups of $G$ (see Table 1).

\footnotesize
\begin{table}[htbp]
\begin{center}
\caption{MLSs for Sporadic Groups}\label{tbl:MLS}
\scalebox{0.8}[0.85]{
\begin{tabular}{|c|c|c|}
\hline
          G    & $H$
               & $G_w$       \\
    \cline {1-3}

   $Co_1$     & $H=ABCD$
              & $G_w= 2^{11}: M_{24}$  \\

              & $|A|=3^6,|B|=5^3, |C|=7,|D|=13$
              &    \\
\cline{1-3}

   $Co_2$      & $H=ABC$
               & $G_w= 2^{10}: M_{22}:2$ \\

               & $|A|=3^4,|B|=5^2, |C|=23$
               &   \\

 \cline{1-3}
  $Fi_{22}$    & $H=ABCD$
               & $G_w= PSU_6(2)$   \\
               & $|A|=2,|B|=3^5, |C|=5,|D|=13$
               &                       \\

 \cline{1-3}

  $Fi_{23}$      & $H=ABC$
                 & $G_w= 2.Fi_{22}$ \\

                 & $|A|=3^4,|B|=17, |C|=23$
                 &    \\

   \cline{1-3}

   $Fi_{24}'$    & $H=ABCD$
                 & $G_w= Fi_{23}$ \\

                 & $|A|=2^3,|B|=3^3, |C|=7^2,|D|=29$
                 &  \\
   \cline{1-3}

   $Th$         &  $H=ABCDEF$
                &  $G_w= ^3D_4(2):3$ \\

                &  $|A|=2^3,|B|=3^5, |C|=5^3,|D|=7,|E|=19, |F|=31$
                &                 \\
    \cline{1-3}

    $HN$        & $H=ABCD$
                & $G_w= A_{12}$ \\

                & $|A|=2^6,|B|=3, |C|=5^5,|D|=19$
                &                               \\

   \cline{1-3}

    $B$         & $H=ABCDEF$

                & $G_w= 2.^2E_6(2):2$ \\

                & $|A|=2^3,|B|=3^4, |C|=5^4,|D|=23,|E|=31, |F|=47$
                &                 \\
    \cline{1-3}

    $M$        &  $H=ABCDEFGH$
               &  $G_w=2.B$ \\
               &  $|A|=2^5,|B|=3^7, |C|=5^3,|D|=11,|E|=13^2, |F|=41, |G|=59, |H|=71$
               &  \\

     \cline{1-3}

    $O'N$       &  $H=ABCD$
                &  $G_w= PSL_3(7):2$   \\

                &  $|A|=2^2,|B|=3^2, |C|=11,|D|=31$
                &    \\
    \cline{1-3}

   $Ly$         &  $H=ABCDE$
                &  $G_w= G_2(5)$  \\

                &  $|A|=2^2,|B|=3^4, |C|=11,|D|=37,|E|=67$
                &    \\
    \cline{1-3}

    $J_3$       &  $H=ABCD$
                &  $G_w= 3 \times (3 \times A_6) : 2$  \\

                &  $|A|=2^2,|B|=3^2, |C|=17,|D|=19$
                &    \\
    \cline{1-3}

   $J_4$       &  $H=ABCDE$
                &  $G_w= 2^{11}:M_{24}$  \\

                &  $|A|=11^2,|B|=29, |C|=31,|D|=37,|E|=43$
                &    \\
    \cline{1-3}

%
%
%
%
%
%
\hline
\end{tabular}
}
\end{center}
\end{table}
\normalsize

\section{Preliminaries}

\subsection{Minimal Logarithmic Signature}

\begin{definition}[Logarithmic Signature (LS)]\cite{GS02,LT05}

Let $G$ be a finite group. Let $\alpha =[A_1, \cdots, A_s]$ be a sequence of ordered subsets $A_i$ of $G$ such that $A_i = [\alpha_{i1}, \cdots, \alpha_{ir_i}]$ with $\alpha_{ij} \in G$, $(1 \leq j \leq r_i)$. If $|G|=r_1\cdot r_2\cdots r_s$ and each
$g \in G$ is uniquely represented as a product
 \begin{center}
  $g = \alpha_{1j_1}\cdots \alpha_{sj_s}$
\end{center}
 with $\alpha_{ij_i}\in A_i (1 \leq i \leq s)$, then $\alpha$ is called a logarithmic signature(LS) for $G$.
\end{definition}

The sequences $A_i$ are called the blocks of $\alpha$, the length of $\alpha$ is defined to be $l(\alpha) =\underset{i=1}{\overset{s}\sum}
 r_i$. let $|G| =\underset{j=1}{\overset{k}\prod}p_{j}^{a_j}$ be the prime power
decomposition of $|G|$ and $\alpha = [A_1, A_2,\ldots, A_s]$ be an LS for $G$. From \cite{GS02}, we can see that $l(\alpha) \geq \underset{j=1}{\overset{k}\sum} a_{j} p_j$.

\begin{definition}[ Minimal Logarithmic Signature (MLS)]\cite{LT05}
Let $\alpha$ be a logarithmic signature for a finite
group $G$, if $l(\alpha) = \underset{j=1}{\overset{k}\sum} a_{j} p_j$, then $\alpha$
is called a minimal logarithmic signature (MLS).
\end{definition}

\begin{lemma}\cite{H04,NN11} \label{le:1}
If $G (G|X)$ contains cyclic subgroups(sets) $A_1,\cdots,A_n$  and $G_w$  such that
\begin{itemize}
  \item [(i)] $|G| = |A_1 |\cdots|A_n|\cdot|G_w|$
  \item [(ii)] $A_i$ has an MLS for all $i$ and $G_w$ has an MLS
  \item [(iii)] $G =A_1\cdots A_n\cdot G_w$ ($H= A_1\cdots A_n \subseteq G$ is a sharply transitive set on $X$ with respect to $w \in X$)
\end{itemize}
 then $[A_1,A_2,\cdots,A_n,G_w]$ is an LS for $G$ and $G$ has an MLS.

\end{lemma}

\begin{lemma}\cite{NNM10,NN11}\label{le:2}
If $G$ is a solvable group, then $G$ has an MLS.
\end{lemma}

\begin{lemma}\cite{NNM10,NN11}\label{le:3}
Let $G$ be a finite group and $x \in G$ be an element of order $t$. For $s \in N, s \leq t$, let $S = \{x^i | 0 \leq i < s\}$ be a cyclic set. Then $S$ has an MLS $\alpha=[A_1,A_2,\cdots, A_k]$ satisfying the following condition:

\begin{center}
      There exist a list $[j_i,j_2,\cdots, j_k]$ such that $x^{j_i}\in A_i, 1\leq i\leq k$, $\underset{i=1} {\overset{k} \sum} j_i <s$.
\end{center}
\end{lemma}

\begin{lemma}\cite{NNM10}\label{le:4}
Let $G$ be a finite group and $[A_1, \cdots, A_r]$ be an LS for $G$. If there exists an MLS for each
subset $A_j$ ($1 \leq j \leq r$), then $G$ has an MLS.
\end{lemma}

\subsection{Sporadic Groups and Related Algebras}

From \cite{W09}, we can see that the sporadic groups come in four classes, three consecutive levels plus the Pariahs. The levels are Mathieu's ( $M_{11}$, $M_{12}$, $M_{22}$, $M_{23}$ and $M_{24}$), Leech's ($Co_1$, $Co_2$, $Co_3$, $HS$, $McL$, $J_2$, $Suz$) and Monster's ($M$, $B$, $Fi_{22}$, $Fi_{23}$, $Fi'_{24}$, $HN$, $Th$, $He$), plus 6 Pariah groups ($Ru$, $O'N$, $Ly$, $J_1$, $J_3$, $J_4$). So there are 26 sporadic groups.

 If $\{e_i \}(1\leq i\leq 24)$ are 24 linearly independent vectors in $R^{24}$, then the points
 $x = \sum_{i=1}^{24} n_ie_i(n_i\in Z)$ form a lattice. If the lattice $\Lambda_{24}$ has the following list of properties \cite{W09}:

\begin{enumerate}
\item[(i)]   it can be generated by the columns of a certain $24\times24$ matrix with determinant 1.
\item[(ii)]  the square of the length of any vector in $\Lambda_{24}$ is an even integer.
\item[(iii)] The length of any non-zero vector in $\Lambda_{24}$ is at least 2.
\end{enumerate}
then $\Lambda_{24}$ is called Leech Lattice.

Actually, Leech lattice are  Z-linear combinations of three types of $(\pm4, \pm4, 0^{22})$, $(2^8,0^{16})$ and
$(-3,1^{23})$ \cite{W09}.Leech Lattice plays an important role in the subsequent construction of the Conway groups. We also need utilize it to construct MLSs for the Conway groups.

The Parker's Loop $\mathbb{P}$ was discovered by Richard Parker. This loop $\mathbb{P}$
 is not a group, but behaves rather like the units $\{\pm1, \pm i_0,\cdots, \pm i_6\}$ of
the octonions \cite{W09}, which is a non-associative
double cover of an elementary abelian group of order 8.
It can be checked that Parker’s loop is a non-associative inverse loop. Here an
inverse loop is a set with a binary operation, an identity element and an inverse
map, satisfying $x1 = x = 1x$, $(x^{-1})^{-1} = x$ and $x^{-1}(xy) = y = (yx)x^{-1}$.

The algebra of $n \times n$ matrices is defined by the well-known matrix product,
which is associative but non-commutative as long as $n > 1$. Meanwhile, we can derive
a commutative Jordan product which is defined by $A \circ B = \frac{1}{2} (AB + BA)$.
Then given two matrices $A$ and $B$, the corresponding Griess product is 4 times the Jordan product:
$A \ast B = 2(AB + BA)$.

\section{Main Results}

\subsection{MLS for $Co_1$ and $Co_2$}

 John H. Conway found in 1968 the automorphism group of the Leech lattice $\Lambda_{24}$ with an enormous group of order

 \begin{center}
$|Aut(\Lambda_ {24})|=|Co_0| = 2^{22} \cdot 3^9 \cdot 5^4 \cdot7^2\cdot 11 \cdot13 \cdot 23$
 \end{center}

This ``zero degree" Conway group is not simple, but it has a simple quotient, named $|Co_1|= |Co_0/Z_2|$, where $Z_2$ is the center of order 2 generated by -1. Two other simple groups $Co_2$ and $Co_3$ were also discovered by Conway using different stabilizers of $Co_0$.

From \cite{W09}, we can see that $Co_1$, $Co_2$ and $Co_3$ are the corresponding stabilizer of a vector of norm 8, norm 6 and norm 4, respectively. Then the corresponding orders are:
\begin{eqnarray*}
|Co_1| &=& 2^{21}.3^9 . 5^4. 7^2. 11. 13.  23\\
|Co_2| &=& 2^{18}. 3^6 . 5^3. 7 . 11 . 23\\
|Co_3| &=& 2^{10}. 3^7 . 5^3 . 7 . 11 . 23
\end{eqnarray*}

Actually, Holmes \cite{H04} proved the existence of MLS for $Co_3$. Therefore, we only consider $Co_1$ and $Co_2$.

\begin{theorem} \label{th:3.1}
 $Co_1$ has an MLS.
\end{theorem}

\begin{proof}
From \cite{W09}, we can see that the number of vectors of norm 8 in the Leech lattice is 398,034,000. Meanwhile,
all the vectors can be divided into 48 classes \cite{W09}. So there are $\frac{398,034,000}{48}=8,292,375=3^6.5^3.7.13$ vectors in each class $w$. Furthermore, the stabilizer $G_w$ of a class $w$ is a semi-direct product $2^{11}:M_{24}$. Then from Sylow Theorem, there are corresponding $3^6$-order, $5^3$-order, 7-order, 13-order subgroups in $Co_1$, where $3^6$-order group and $5^3$-order group are elementary abelian groups \cite{W09}. Besides, all the subgroups above intersect 1, so in this sense, we can divide $Co_1$ into five parts: $A$, $B$, $C$, $D$ and $G_w$. From Lemma \ref{le:1}, $Co_1$ has an LS $[A, B, C, D, G_w]$. Meanwhile, $A$, $B$, $C$, $D$ are all solvable groups and they consist of fixed direct products of groups with prime orders, so all of them have MLSs \cite{S04,BH12}; the group of order $2^{11}$ is an elementary abelian 2-group which also has an MLS \cite{S04,BH12}; $M_{24}$ also has an MLS \cite{GRS03,H04}. So from Lemma \ref{le:2}, Lemma \ref{le:3} and Lemma \ref{le:4}, $Co_1$ has an MLS.\qed

\end{proof}

\begin{theorem} \label{th:3.2}
 $Co_2$ has an MLS.
\end{theorem}
\begin{proof}

From \cite{W09}, the order of $Co_2$ is $ 2^{18}. 3^6 . 5^3. 7 . 11 . 23$. When we take the fixed vector $w$ of
$Co_2$ to be $(4,-4,0^{22})$, we can get a stabilizer subgroup $G_w= 2^{10}:M_{22}:2$ with index $46,575=3^4.5^2.23$ \cite{CCNPW85}. Also from Sylow Theorem, there are corresponding $3^4$-order, $5^2$-order and $23$-order subgroups in $Co_2$ and all the subgroups intersect 1. Then we can divide $Co_2$ into $A$, $B$, $C$ and $G_w$. From Lemma \ref{le:1}, $Co_2$ has an LS $[A, B, C, G_w]$. Actually, $A$, $B$, $C$ are solvable groups and consist of fixed direct products of groups with prime orders, the elementary abelian 2-group of order $2^{10}$ is also a solvable group. $M_{22}:2$ is a double cover of $M_{22}$ which has an MLS \cite{GRS03}. Consequently, from Lemma \ref{le:4}, $Co_2$ has an MLS.\qed

\end{proof}

\subsection{MLS for some Monster's Groups}

 We will firstly introduce some simple Monster's Groups -- the Fischer Groups. We need take advantage of the group $\Omega_7(3)$ which is the commutator subgroup of orthogonal group $O_7(3)$ to describe the construction of $Fi_{22}$. Let $\mathbb{B}=\{x_1,x_2,\cdots, x_7\}$ be the basis of $\Omega_7(3)$, the corresponding quadratic form is $Q(x_1,x_2,\cdots,x_7)= \underset{i=1}{\overset{6}\sum} x_i^2-x_7^2$ \cite{W09}. We can now define $Fi_{22}$ to be the group generated by the 3510 3-transpositions given in \cite{W09}. In order to compute the order of $Fi_{22}$, we shall determine the vertex stabiliser, which turns out to be a double cover of $PSU_6(2)$. So we can see that \cite{W09}

\begin{center}

$|Fi_{22}| = 3510.|2\cdot PSU_6(2)|= 2^{17}.3^9.5^2.7.11.13$

\end{center}

Meanwhile, there is a a double cover $2.Fi_{22}$ which is also generated by the mentioned 3510 3-transpositions in \cite{W09}. Then the group $Fi_{23}$ generated by the 31671 3-transpositions is a group of
order $31671.|2.Fi_{22}|$. Thus \cite{W09}

\begin{center}

$|Fi_{23}| = 2^{18}.3^{13}.5^2.7.11.13.17.23$.

\end{center}

 From \cite{W09}, the group $Fi_{24}$ is of order $306936.2.|Fi_{23}|$. Although $Fi_{24}$ is not simple, it has a simple subgroup $Fi'_{ 24}$ with index 2 and order \cite{W09}

\begin{center}

$|Fi'_{24}| = 2^{21}.3^{16}.5^2.7^3.11.13.17.23.29$.

\end{center}


The Monster group $M$ was finally constructed by Griess in 1980 as an
automorphism group of a remarkable commutative but non-associative Griess algebra with 196,884
dimensions \cite{R06}. Besides, Conway also utilize Parker’s Loop $\mathbb{P}$ to compute the order of
$M$. The Monster is so called largely because of its enormous size. Its order is \cite{W09}

\begin{center}

$|M|=2^{46}.3^{20}.5^9.7^6.11^2.13^3.17.19.23.29.31.41.47.59.71$.

\end{center}

We have seen that the Baby Monster group $B$ has a double cover $2.B$ which is a
subgroup of the Monster. Hence we can see that the order of the Baby Monster group is \cite{W09}

\begin{center}

$|B| =2^{41}.3^{13}.5^{6}.7^{2}.11.13.17.19.23.31.47$.

\end{center}

Besides, there are other 3 simple subgroups of $M$: Thompson sporadic simple group, found by John G. Thompson (1976) and constructed by Smith (1976) with order $|Th| =2^{15}.3^{10}.5^3.7^2.13.19.31$; Harada-Norton group $HN$, found by Harada (1976) and Norton (1975), with order $|HN| =2^{14}.3^6.5^6.7.11.19$; Held group $He$, found by Dieter Held (1969), with order $|He| =2^{10}.3^3.5^2.7^3.17$.
 Actually, Holmes \cite{H04} proved the existence of MLSs for $He$, so we only consider the existence of MLS for the remained Monster's Groups.

\begin{theorem} \label{th:3.3}
$Fi_{22}$ has an MLS.
\end{theorem}
\begin{proof}
As described above, let $w=\langle x_1\rangle$，then the point stabilizer $G_w$ is the double cover of the group $PSU_6(2)$ with index $3,510=2.3^3.5.13$. Also from Sylow Theorem, $Fi_{22}$ has corresponding 2-order, $3^3$-order, 5-order and 13-order subgroups. All the subgroups intersect 1. So $Fi_{22}$ can be divided into $A$, $B$, $C$, $D$ and $G_w$. Then from Lemma \ref{le:1}, $[A, B, C, D, G_w]$ is an LS for $Fi_{22}$. Actually, $A$, $B$, $C$, $D$ consist of fixed direct products of groups with prime orders and are all solvable groups, $PSU_6(2)$ has an MLS \cite{LT05}. Therefore, from Lemma \ref{le:4}, $Fi_{22}$ has an MLS .\qed

\end{proof}

\begin{theorem} \label{th:3.4}
$Fi_{23}$ has an MLS.
\end{theorem}
\begin{proof}

与情形类似，Let $\mathbb{B}=\{x_1,x_2,\cdots,x_8 \}$ be the basis of $\Omega_8(3)$, $w=\langle x_1\rangle$, then the point stabilizer $G_w$ of $Fi_{23}$ is the double cover of the $Fi_{22}$ denoted by $G_w=2\cdot Fi_{22}$. Meanwhile, the index of $G_w$ in $Fi_{22}$ is $31,671=3^4.17.23$. Then from Sylow Theorem, $Fi_{23}$ has the corresponding $3^4$-order, 17-order and 23-order  subgroups and all the subgroups intersect identity element 1. So $Fi_{23}$ can be divided into $A$, $B$, $C$ and $G_w$. Then from Lemma \ref{le:1}, $[A, B, C, G_w]$ is an LS for $Fi_{23}$. Also, $A$, $B$, $C$ consist of fixed direct products of groups with prime orders and are all solvable groups, so they all have corresponding MLSs \cite{S04,BH12}. Furthermore, from Theorem \ref{th:3.3}, the double cover of the $Fi_{22}$ also has corresponding MLS. Hence, from Lemma \ref{le:4}, $Fi_{23}$ has an MLS.\qed

\end{proof}

\begin{theorem} \label{th:3.5}
$Fi'_{24}$ has an MLS.
\end{theorem}
\begin{proof}

From \cite{W09}, the simple group $Fi'_{24}$ is a subgroup of $Fi_{24}$ with index 2 and the point stabilizer
$G_w$ of $Fi'_{24}$ is $Fi_{23}$ with index $306,936=2^3.3^3.7^2.29$. Also from Sylow Theorem, there are $2^3$-order, $3^3$-order, $7^2$-order and 29-order subgroups in $Fi'_{24}$ and all these subgroup share the only common identity element 1. Actually, we can appropriately choose $A$, $B$, $C$ and $D$ with the corresponding orders above, then joint them with $G_w$ to construct $Fi'_{24}$. Hence, from Lemma \ref{le:1},               $[A, B, C, D, G_w]$ is an LS for $Fi_{24}$. $A$, $B$, $C$ and $D$ are all solvable groups, so they all have MLSs. From Theorem \ref{th:3.4}, $G_w= Fi_{23}$ also has an MLS. Thus, from Lemma \ref{le:4}, $Fi'_{24}$ has an MLS.\qed

\end{proof}

\begin{theorem} \label{th:3.6}
$Th$ has an MLS.
\end{theorem}
\begin{proof}

From \cite{W09}, Thompson group $Th$ is a subgroup of automorphisms of a certain lattice in the 248-dimensional Lie algebra of $E_8(3)$. Meanwhile, let $w=\langle e_1\rangle$, then the point stabilizer $G_w$ is $^3D_4(2):3$ with index $143,127,000=2^3.3^5.5^3.7.19.31$. Then from Sylow Theorem, $Th$ has the corresponding $2^3$-order, $3^5$-order, $5^3$-order, 7-order, 19-order and 31-order subgroups. Meanwhile, all these subgroups only intersect identity element 1. Then we appropriately choose $A$, $B$, $C$, $D$ and $E$ with the corresponding orders above, joint them with $G_w$ to construct the whole $Th$. Hence, from Lemma \ref{le:1}, $[A, B, C, D, E, G_w]$ is an LS for $Th$. Besides, $A$, $B$, $C$, $D$ and $E$ are all solvable groups, $^3D_4(2):3$ also has an MLS due to the existence of MLS for $^3D_4(2)$ \cite{W09}. Consequently, from Lemma \ref{le:4}, $Th$ has an MLS.\qed

\end{proof}

\begin{theorem} \label{th:3.7}
$HN$ has an MLS.
\end{theorem}
\begin{proof}

From \cite{W09}, $|HN| = 2^{14}.3^6.5^6.7.11.19$, the stabilizer $G_w$ is $A_{12}$ with index
$1,140,000= 2^6.3.5^5.19$. Then from Sylow Theorem, there are corresponding $2^6$-order, 3-order, $5^5$-order and 19-order subgroups in $HN$. Besides, they only have the common identity element 1. Hence, we can appropriately choose $A$, $B$, $C$ and $D$ with corresponding orders above, joint them with $G_w$ to construct $HN$. Therefore, from Lemma \ref{le:1}, $[A, B, C, D, G_w]$ is an LS for $HN$. Actually, $A$, $B$, $C$ and $D$ are solvable groups which have corresponding MLSs. $A_{12}$ also has an MLS \cite{M02}. Also, from Lemma \ref{le:4}, $HN$ has an MLS.\qed

\end{proof}

\begin{theorem} \label{th:3.8}
$B$ has an MLS.
\end{theorem}
\begin{proof}
As described as above, $|B| = 2^{41}.3^{13}.5^6.7^2.11.13.17.19.23.31.47$. From \cite{W09}, the stabilizer $2.^2E_6(2):2$ is the subgroup fixing a point of the smallest permutation representation on $13,571,955,000 =
2^3.3^4.5^4.23.31.47$ points. Besides, the corresponding $2^3$-order, $3^4$-order, $5^4$-order, 23-order, 31-order and 47-order subgroups only have common element 1. Then we can appropriately select $A$, $B$, $C$, $D$ and $E$ with corresponding orders above to construct the whole $B$. So $[A, B, C, D, G_w]$ is an LS for $B$. Since
$A$, $B$, $C$, $D$ and $E$ are all solvable groups, $2.^2E_6(2):2$ also has an MLS due to the existence of MLS for $^2E_6(2)$ \cite{W09}, then from Lemma \ref{le:4}, $B$ has an MLS .\qed

\end{proof}

\begin{theorem} \label{th:3.9}
$M$ has an MLS.
\end{theorem}
\begin{proof}
Through utilizing 196884- dimensions Griess algebra as \cite{W09}, we can get $|M| = 2^{46}.3^{20}.5^9.7^6.11^2.13^3.17.19.23.29.31.41.47.59.71$. Also let $w=\langle x_1\rangle$, then the stabilizer $G_w$ is the double cover of $B$ and its index is $2^5.3^7.5^3.11.13^2.41.59.71$. Then there are corresponding $2^5$-order, $3^7$-order, $5^3$-order, $11$-order, $13^2$-order, $41$-order, $59$-order, $71$-order subgroups in $M$ and they only intersect 1. Then with appropriately selected $A$, $B$, $C$, $D$, $E$, $F$, $G$ and $H$ with corresponding orders mentioned above, we can construct the whole $M$. Then $[A, B, C, D, E, F, G, H, G_w]$ is an LS for $B$. Since $A$, $B$, $C$, $D$, $E$, $F$, $G$ and $H$ are all solvable groups, $2.B$ also has an MLS due to the existence of MLS for $B$ from Theorem \ref{th:3.8}, then from Lemma \ref{le:4}, $M$ has an MLS.\qed

\end{proof}

\subsection{MLS for some Pariah Groups}

There are 6 Pariah groups: $Ru$, $O'N$, $Ly$, $J_1$,
$J_3$, $J_4$. Since the existence of MLSs for $J_1$ and $Ru$ has been proved in \cite{H04}, we only consider the existence of MLSs for the remained four types of groups - $O'N$, $Ly$, $J_3$ and $J_4$.

%
%
%
%
%
%
%
%
%
%
%
%
%
%
%
%
%
%
%
%
%
%
%
%
%
%
%
%
%
%
%
%
%
%
%

\begin{theorem} \label{th:3.10}
$J_3$ has an MLS.
\end{theorem}
\begin{proof}
From \cite{W09}, $|J_3| = 2^7.3^5.5.17.19$. Let $w=\langle e_0 \rangle$ be the 1-space, then the stabilizer $G_w=3\times(3\times A_6):2$ with index $23,256=2^2.3^2.17.19$. Meanwhile, there are corresponding $2^2$-order, $3^2$-order, 17-order and 19-order subgroups in $J_3$. So in this sense, we can select suitable subgroups $A$, $B$, $C$ and $D$ with corresponding orders above, then joint them with $G_w$ to construct the whole $J_3$. Thus, from Lemma \ref{le:1}, $[A,B,C,D,G_w]$ is an LS for $J_3$. Also from Lemma \ref{le:2} and Lemma \ref{le:3}, $A$, $B$, $C$ and $D$ have the corresponding MLSs. Due to the existence of MLS for $A_6$ \cite{M02,GRS03,LT05}, so from Lemma \ref{le:4}, $J_3$ has an MLS.\qed

\end{proof}

\begin{theorem} \label{th:3.11}
$J_4$ has an MLS.
\end{theorem}
\begin{proof}
Also from \cite{W09}, $|J_4| = 2^{21}.3^3.5.7.11^3.23.29.31.37.43$. there is an orbit of $173,067,389=11^2.29.31.37.43$ vectors on which the group acts transitively, with point stabiliser $G_w =2^{11}:M_{24}$. Meanwhile, there are corresponding $11^2$-order, 29-order, 31-order, 37-order, 43-order subgroups in $J_4$. So in this sense, we can select suitable subgroups $A$, $B$, $C$, $D$ and $E$ with corresponding orders above. then joint them with $G_w$ to construct the whole $J_4$. Hence, from Lemma \ref{le:1}, $[A,B,C,D, E, G_w]$ is an LS for $J_4$. Also from  Lemma \ref{le:2} and  Lemma \ref{le:3}, $A$, $B$, $C$, $D$ and $E$ have the corresponding MLSs. Due to the existence of MLS for $M_{24}$ \cite{GRS03,LT05}, from Lemma \ref{le:4}, $J_4$ has an MLS .\qed

\end{proof}

\begin{theorem} \label{th:3.12}
$O'N$ has an MLS.
\end{theorem}
\begin{proof}
From \cite{W09}, we can see that $|O'N| = 2^9.3^4.5.7^3.11.19.31$. Then from \cite{W09}, The point stabiliser $G_w= PSL_3(7):2$ is the subgroup which fix a point of the smallest permutation representation on $122,760= 2^2.3^2.11.31$ points. Also from Sylow Theorem, $O'N$ has corresponding $2^2$-order, $3^2$-order, $11$-order and 31-order subgroups and they only intersect identity element 1. Then we appropriately select $A$, $B$, $C$ and $D$ with corresponding orders above, joint them with $G_w$ to construct $O'N$. So the divided parts form an LS. Hence, $[A,B,C,D,G_w]$ is an LS for $O'N$. Since $A$, $B$, $C$ and $D$ are solvable groups, $G_w= PSL_3(7):2$
also has an MLS \cite{NNM10,NN11}. Therefore, from Lemma \ref{le:4}, $O'N$ has an MLS.

\end{proof}

\begin{theorem} \label{th:3.13}
$Ly$ has an MLS.
\end{theorem}
\begin{proof}
From \cite{W09}, $|Ly| = 2^8.3^7.5^6.11.31.37.67$, the stabilizer $G_w=G_2(5)$ is the maximum subgroups of
$Ly$ which is the smallest permutation representation on $8,835,156=2^2.3^4.11.37.67$ points. Also from Sylow Theorem, $Ly$ has the corresponding $2^2$-order, $3^4$-order, 11-order, 37-order and 67-order subgroups. Then we appropriately choose $A$, $B$, $C$, $D$ and $E$ so that they intersect only element 1. So in this way, $Ly$ can be divided into 7 parts. Consequently, $[A,B,C,D,E,G_w]$ is an LS for $Ly$. Since $A$, $B$, $C$, $D$ and $E$ are all solvable, $G_2(5)$ also has corresponding MLS \cite{LT05}. Hence, from Lemma \ref{le:4}, $Ly$ has an MLS.\qed

\end{proof}

\section{Conclusions}

We utilize Sylow Theorem and stabilizers of the corresponding sporadic groups to construct MLSs for 13 types of sporadic groups. Now, we can get the conclusion that all sporadic groups have MLSs. Meanwhile, our methods can be used to construct MLSs for other finite simple groups.

\section*{Acknowledgements}
This work is partially supported by the National Natural Science Foundation of China (NSFC) (Nos. 61070251, 61103198, 61121061) and the NSFC A3 Foresight Program (No. 61161140320).












\begin{thebibliography}{99}


\bibitem{H70}
Hestenes M.D.: Singer groups. Canad. J. Math. 22, 492-513 (1970).



\bibitem{M86}
Magliveras S.S.:A cryptosystem from logarithmic signatures of
finite groups. In: Proceedings of the 29th Midwest Symposium on
Circuits and Systems, pp. 972-975. Elsevier Publishing Company,
Amsterdam (1986).



\bibitem{MM92}
Magliveras S.S., Memon N.D: Algebraic properties of cryptosystem
PGM. J. Cryptol. 5, 167-183 (1992).




\bibitem{GS02}
González Vasco M.I., Steinwandt R.: Obstacles in two public key
cryptosystems based on group factorizations. Tatra Mt. Math. Publ.
25, 23-37 (2002).

\bibitem{M02}
Magliveras S.S.: Secret and public-key cryptosystems from group
factorizations. Tatra Mt. Math. Publ. 25, 11-22 (2002).

\bibitem{MST02}
Magliveras S.S., Stinson D.R., van Trung T.: New approaches to
designing public key cryptosystems using one-way functions and
trapdoors in finite groups. J. Cryptol. 15, 285-297 (2002).

\bibitem{GRS03}
González Vasco M.I., Rotteler M., Steinwandt R.: On minimal length factorizations of finite groups. Exp. Math. 12, 1-12 (2003).



\bibitem{CD04}
Cossidente A., De Resmini M.J.: Remarks on singer cyclic groups and their normalizers. Des. Codes Cryptogr. 32, 97-102 (2004).

\bibitem{H04}
Holmes P.E.: On minimal factorisations of sporadic groups. Exp. Math. 13, 435-440 (2004).





\bibitem{LT05}
Lempken W., van Trung T.: On minimal logarithmic signatures of
finite groups. Exp. Math. 14, 257-269 (2005).




\bibitem{W09}
Wilson R.A.: The finite simple groups. Graduate Texts
in Mathematics, vol 251. Springer-Verlag, London (2009).


\bibitem{BPS09}
Babai L., Palfy P.P., Saxl J.: On the number of p regular elements in finite simple groups. LMS J. Comput. Math. 12, 82-119 (2009).


\bibitem{LMTW09}
Lempken W., Magliveras S.S., van Trung T., Wei W.: A public key
cryptosystem based on non-abelian finite groups. J. Cryptol. 22,
62-74 (2009).



\bibitem{NNM10}
S.~Nikhil, S.~Nidhi, S.S.~Magliveras: Minimal logarithmic
signatures for finite groups of lie type. Des. Codes Cryptogr. 55,
243-260 (2010).


\bibitem{NN11}
S. Nikhil, S. Nidhi: Minimal logarithmic signatures for
classical groups. Dec.Codes Cryptogr.60, 183-195 (2011)

\bibitem{MST12}
 Pascal Marquardt, Pavol Svaba, Tran van Trung: Pseudorandom number generators based on random covers for finite groups. Des. Codes Cryptogr.64, 209--220 (2012)

\bibitem{CCNPW85}

J. Conway, R. Curtis, S. Norton, R. Parker and R. Wilson: Atlas of finite groups. Clarendon Press,
Oxford (1985).
\bibitem{R06}
Ronan M.: Symmetry and the Monster. One of the greatest quests of mathematics, Oxford University Press,
Oxford, 2006.

\bibitem{BH12}
B. Baumeister, J.H.de Wiljes: Aperiodic logarithmic signatures. J. Math. Cryptol. 6, 21-37,
(2012)

\bibitem{S04}
S$\acute{a}$andor Szab$\acute{o}$, Topics in Factorization of Abelian Groups, Birkh$\ddot{a}$user Verlag, Basel - Boston- Berlin (2004).




\end{thebibliography}
\end{document}